%% file: paper.tex
\documentclass[copyright,creativecommons]{eptcs}

\providecommand{\event}{QAPL 2012}

\title{Differential Privacy for Relational Algebra: Improving the
  Sensitivity Bounds via Constraint Systems\thanks{This work has been
    partially supported by the project ANR-09-BLAN-0169-01 PANDA}}
\author{Catuscia Palamidessi
  \institute{INRIA and LIX, Ecole Polytechnique, France}
  \email{catuscia@lix.polytechnique.fr}
  \and
  Marco Stronati
  \institute{Universit\`a di Pisa, Italy}
  \email{marco.stronati@gmail.com}
}
\def\titlerunning{Differential Privacy for Relational Algebra}
\def\authorrunning{C. Palamidessi \& M. Stronati}
\date{}
\input{style}

\DeclareMathAlphabet{\mathcal}{OT1}{pzc}{m}{n}

\begin{document}
\maketitle

\begin{abstract}
  Differential privacy is a modern approach in
  privacy-preserving data analysis to control the amount of information
  that can be inferred about an individual by querying a database. The
  most common techniques are based on the introduction of probabilistic
  noise, often defined as a Laplacian parametric on the sensitivity of
  the query. In order to maximize the utility of the query, it is
  crucial to estimate the sensitivity as precisely as possible.\\
  In this paper we consider relational algebra, the classical language
  for queries in relational databases, and we propose a method for
  computing a bound on the sensitivity of queries in an intuitive and
  compositional way. We use constraint-based techniques to accumulate
  the information on the possible values for attributes provided by the
  various components of the query, thus making it possible to compute
  tight bounds on the sensitivity.
\end{abstract}

\section{Introduction}
Differential privacy~\cite{Dwork:06:ICALP,DBLP:dblp_conf/tamc/Dwork08,
Dwork:11:CACM,Dwork:09:STOC} is a recent approach addressing the
privacy of individuals in data analysis on statistical databases. In
general, statistical databases are designed to collect global
information in some domain of interest, while the information about
the particular entries is supposed to be kept confidential.
Unfortunately, querying a database might leak information about an
individual, because the presence of her record may induce the query to
return a different result.

To illustrate the problem, consider for instance a database of people
affected by a certain disease, containing data such as age, height, etc.
Usually the identity of the people present in the database is supposed
to be secret, but if we are allowed to query the database for the
number of records which are contained in it, and for -- say -- the
average value of the data (height, age, etc.), then one can infer the
precise data of the last person entry in the database, which poses a
serious threat to the disclosure of her identity as well.

To avoid this problem, one of the most commonly used methods consists
in introducing some \emph{noise} on the answer. In other words,
instead of giving the \emph{exact} answer the curator gives an
\emph{approximated} answer, chosen randomly according to some
probability distribution.

Differential privacy measures the level of privacy provided by such a
randomized mechanism by a parameter $\epsilon$: a mechanism $\cal K$
is $\epsilon$-differentially private if for every pair of
\emph{adjacent} databases $R$ and $R'$ (i.e. databases which which
differ for only one entry), and for every property $\cal P$, the
probabilities that ${\cal K}(R)$ and ${\cal K}(R')$ satisfy $\cal P$
differ at most by the multiplicative constant $e^\epsilon$.

The amount of noise that the mechanism must introduce in order to
achieve $\epsilon$ differential privacy depends on the so-called
\emph{sensitivity} of the query, namely the maximum distance between
the answers on two adjacent databases. For instance, one of the most
commonly used mechanisms, the \emph{Laplacian}, adds noise to the
correct answer $y$ by reporting an approximated answer $z$ according
to the following probability density function:
\[ P_y(z)= c\, e^{-\frac{|y-z|}{\Delta f}\epsilon} \]
where $\Delta f$ is the sensitivity of the query $f$, and $c$ is a
normalization factor. Clearly, the higher is the sensitivity, the
greater the noise, in the sense that the above function is more
``flat'', i.e. we get a higher probability of reporting an answer very
different from the exact one.

Of course, there is a trade off between the privacy and the
\emph{utility} of a mechanism: the more noise a mechanism adds, the
less precise the reported answer, which usually means that the
result of querying the database becomes less useful -- whatever the
purpose.

For this reason, it is important to avoid adding excessive noise: one
should add only the noise strictly necessary to achieve the desired
level of differential privacy. This means that the sensitivity of the
query should be computed as precisely as possible. At the same time,
for the sake of efficiency it is desirable that the computation of the
sensitivity is done \emph{statically}. Usually this implies that we
cannot compute the precise sensitivity, but only approximate it from
above. The goal of this paper is to explore a \emph{constraint-based
methodology} in order to compute strict upper bounds on the
sensitivity.

The language we chose to conduct our analysis is \emph{relational
algebra} \cite{Codd:70:CACM,Codd:72:BOOK}, a formal and well defined
model for relational databases, that is the basis for the popular
Structured Query Language (SQL,~\cite{Beaulieu:09:BOOK}). It consists
in a collection of few operators that take relations as input and
return relations as output, manipulating rows or columns and computing
aggregation of values.

Sensitivity on aggregations often depends on attribute ranges, and
these restrictions can be exploited to provide better bounds. To this
purpose, we extend mechanisms already in place in modern database
systems: In RDBMS (Relational Data Bases Management Systems)
implementations, during the creation of a relation, it is possible to
define a set of constraints over the attributes of the relation, to
further restrict the type information. For instance:
\vskip -1em
\begin{center} 
  {\small \tt{Persons\{(Name, String)(Age, Integer)\}
      \{Age > 0 $\wedge$ Age < 120)\}}}
\end{center}
\vskip .5em
refines the type integer used to express the age of a person in the
database, by establishing that it must be a positive value smaller
than $120$.

Constraints in RDBMS can be defined on single attributes (\emph{column
constraints}), or on several attributes (\emph{table constraints}),
and help define the structure of the relation, for example by stating
whether an attribute is a primary key or a reference to an external
key. In addition, so called \emph{check constraints} can be defined,
to verify the insertion of correct values. In the example above, for
instance, the constraint would avoid inserting an age of, say,
200. Check constraints are particularly useful for our purposes
because they restrict the possible values of the attributes, thus
allowing a finer analysis of the sensitivity.

\paragraph{Contribution}
Our contribution is twofold:
\begin{enumerate}
\item we propose a method to compute a bound on the sensitivity of a
  query in relational algebra in a compositional way, and
\item we propose the use of constraints and constraint solvers to
  refine the method and obtain strict bounds on queries which have
  aggregation functions at the top level.
\end{enumerate}

\paragraph{Plan of the paper}
Next section recalls some preliminary notions about relational
databases and differential privacy. Section~\ref{sec:constraints}
introduces a constraint system and the idea of carrying along the
information provided by the constraints as we analyze the query.
Section~\ref{sec:extended} proposes a generalization of differential
privacy and sensitivity to generic metric spaces. This generalization
will be useful in order to compute the sensitivity of a query in a
compositional way.
Sections~\ref{sec:row-op}, \ref{sec:att-op} and \ref{sec:agg} analyze
the sensitivity and the propagation of constraints for the various
operators of relational algebra.
Finally Section~\ref{sec:glob-sen} proposes a method to compute a
sensitivity bound on the global query, and shows its correctness and
the improvement provided by the use of
constraints. Section~\ref{sec:related} discusses some related work,
and Section~\ref{sec:conclusion} concludes.
Due to space limitations, in this version we have omitted several
proof. The interested reader can find them in the full online version
of the paper \cite{self}.

\section{Preliminaries}\label{sec:preliminaries}
We recall here some basic notions about relational databases and
relational algebra, differential privacy, and sensitivity.

\subsection{Relational Databases and Relational Algebra}

Relational algebra \cite{Codd:70:CACM,Codd:72:BOOK} can be considered
as the theoretic foundation of database query languages and in
particular of SQL~ \cite{Beaulieu:09:BOOK}. It is based on the concept
of relation, which is the mathematical essence of a (relational)
database, and of certain operators on relations like union,
intersection, projections, filters, etc..
Here we recall the basic terminology used for relational databases,
while the operators will be illustrated in detail in the technical
body of the paper.

A relation (or database) based on a certain schema is a collection of
tuples (or records) of values. The schema defines the types (domain)
and the names (attributes) of these values.

\begin{definition}[Relation Schema]\label{def:relation_schema}
  A relation schema $r(a_1 : D_1, a_2 : D_2, \ldots, a_n : D_n)$
  is composed of the relation name $r$ and a set of attributes $a_1, a_2,
  \ldots, a_n$ associated with the domains $D_1, D_2, \dots, D_n$, respectively.
  We use the notation $dom(a_i)$ to refer to $D_i$.
\end{definition}

\begin{definition}[Relation]
  A relation $R$ on a relation schema $r(a_1 : D_1,a_2 : D_2,\ldots, a_n : D_n)$ 
  is a subset of the Cartesian product $D_1 \times D_2 \times \ldots \times D_n$.
\end{definition}

A relation is thus composed by a set of n-tuples, where each n-tuple
$\tau$ has the form $(d_1, d_2, \ldots, d_n)$ with $d_i \in D_i$. 
Note that $\tau$ can also be seen as a partial function from
attributes to atomic values, i.e. $\tau(a_i) = d_i$.
Given a schema, we will denote the universe of possible tuples by
$\mathcal{T}$, and the set of all possible relations by $\mathcal{R} = 2^{\mathcal{T}}$.

Relational algebra is a language that operates from relations to
relations. Differentially private queries, however, can only return a
value, and for this reason they must end with an aggregation (operator
$\gamma$). Nevertheless it is possible to show that the full power of
relational algebra aggregation can be retrieved.

\subsection{Differential Privacy}
Differential privacy is a property meant to guarantee that the
participation in a database does not constitute a threat for the
privacy of an individual. More precisely, the idea is that a
(randomized) query satisfies differential privacy if two relations
that differ only for the addition of one record are almost
indistinguishable with respect to the results of the query.

Two relations $R,R' \in \mathcal{R}$ that differ only for the addition
of one record are called \emph{adjacent}, denoted by $R\sim
R'$. Formally, $R\sim R'$ iff $R \setminus R' = \{\tau\}$ or viceversa
$R' \setminus R = \{\tau\}$, where $\tau$ is a tuple.

\begin{definition}[Differential privacy \cite{Dwork:06:ICALP}]\label{def:dp}
A randomized function ${\cal K}: {\cal R}\rightarrow Z$ satisfies $\epsilon$-differential
privacy if for all pairs $R,R' \in {\cal R}$, with $R \sim R'$, and all $Y
\subseteq Z$, we have that:
$$ \mathit{Pr}[{\cal K}(R) \in Y] \leq \mathit{Pr}[{\cal K}(R') \in Y] \cdot e^{\epsilon}  $$
where $\mathit{Pr}[E]$ represents the probability of the event $E$.
\end{definition}

Differentially private mechanisms are usually obtained by adding some
random noise to the result of the query. The best results are obtained
by calibrating the noise distribution according to the so-called
sensitivity of the query. When the answers to the query are real
numbers ($\mathbb{R}$), its sensitivity is defined as follows. (We
represent a query as a function from databases to the domain of
answers.)

\begin{definition}[Sensitivity \cite{Dwork:06:ICALP}]\label{def:sensitivity}
Given a query ${Q}: {\cal R}\rightarrow \mathbb{R}$, the sensitivity
of $Q$, denoted by $\Delta_Q$, is defined as:
$$\Delta_Q = \sup_{R\sim R'} \,|\,Q(R) - Q(R')\,|.$$
\end{definition}

The above definition can be extended to queries with answers on
generic domains, provided that they are equipped with a notion of
distance.

\section{Databases with constraints}\label{sec:constraints}

As explained in the introduction, one of the contributions of our paper 
is to provide strict bounds on the sensitivity of queries by 
using constraints.
For an introduction to the notions of constraint, constraint solver,
and constraint system we refer to \cite{Apt:2003:BOOK}.

In this section we define the constraint system that we will use, and
we extend the notion of database schema so to accommodate the
additional information provided by the constrains during the analysis
of a query.

\begin{definition}[Constraint system]
Our constraint system is defined as follows:
  \begin{itemize}
    \setlength{\itemsep}{1pt}\setlength{\parskip}{0pt}\setlength{\parsep}{0pt} 
  \item Terms are constructed from: 
    \begin{itemize}
    \item variables, ranging over the attribute names of the schemas,
    \item constants, ranging over the domains of the schemas,
    \item applications of n-ary functions (e.g. $+,\times$) to n terms.
    \end{itemize}
  \item Atoms are applications of n-ary predicates to n terms. Possible predicates are $\geq,\leq,=, \in$.
  \item Constraints are constructed from: 
    \begin{itemize}
    \item atoms, and
    \item applications of logical operators ($\neg, \wedge, \vee, \equiv$) to constraints.
    \end{itemize}
  \end{itemize}
\end{definition}
We denote the composition of constraint by $\otimes$.  The solutions
of a set of constraints $C$ is the set of tuples that satisfy $C$,
denoted ${\mathit Sol}(C)$.  The relations that can be build from
$sol(C)$ are denoted by $\mathcal{R}(C) = \mathcal{P}(sol(C))$.  The
solutions with respect to an attribute $a$ is denoted $sol(C,a)$.
Namely, $sol(C,a)$ is the projection on $a$ of ${\mathit sol}(C)$.
When the domain is equipped with an ordering relation, we also use
$\mathit{inf}(C,a)$ and $\mathit{sup}(C,a)$ to denote the infimum and
the supremum values, respectively, of $sol(C,a)$.
Typically the solutions and the {\it inf} and {\it sup} values can be
computed automatically using constraint solvers.  Finally we define
the diameter of a constraint C as the maximum distance between the
solutions of C.

\begin{definition}[Diameter]
  The diameter of a constraint $C$, denoted $diam(C)$, is the graph
  diameter of the adjacency graph $(\mathcal{R}(C), \sim)$ of all
  possible relations composed by tuples that satisfy $C$.
\end{definition}
We now extend the classical definition of schema to contain also the
set of constraints.

\begin{definition}[Constrained schema]\label{def:constrained_schema}
  A constrained schema $r(A,C)$ is composed of the relation name $r$,
  a set of attributes $A$, and a set of constraints $C$. A relation on a
  constrained schema is a subset of ${\mathit sol}(C)$. We will use
  $schema(R)$ to represent the constrained relation schema of a relation
  $R$.
\end{definition}

The above definition extends the notion of relation schema
(Definition~\ref{def:relation_schema}): In fact here each $a_i$ can be
seen as associated with $sol(C,a_i)$.
Definition~\ref{def:relation_schema} can then be retrieved by imposing
as only constraints those of the form $a_i\in D_i$.

\begin{example} Consider the constrained schema ${\tt Items}(A,C)$, where
$A = \{{\tt Item}, {\tt Price}, {\tt Cost}\}$, and
$C = \{({\tt Cost} \leq {\tt Price} \leq 1000, 0 < {\tt Cost} \leq 1000)\}$. 
The following $R$ is a possible relation over this schema.
$R$:
\end{example}

\vskip 1em
\begin{tabular}{lr}
  \begin{minipage}{.6\linewidth}
    \small \ttfamily
    \begin{tabular}{ll}
    \\
      Items &\{Item, Price, Cost\} \\
      &\{(Cost $\leq$ Price $\leq$ 1000, 0$<$Cost$\leq$1000)\}
      \\ \\ 
      &${\tt Items}(A,C)$
    \end{tabular}
  \end{minipage}
  &\qquad
  \begin{minipage}{.2\linewidth}
    \small \ttfamily
    \begin{tabular}{lll}
      Item & Price & Cost \\ \hline
      Oil  & 100   & 10   \\
      Salt & 50    & 11   \\
      \\
      &$R$
    \end{tabular}
  \end{minipage}
\end{tabular}

\section{Differential privacy on arbitrary metrics}\label{sec:extended}
The classic notions of differential privacy and sensitivity are meant
for queries defined on $\cal R$, the set of all relations on a given
schema. The adjacency relation induces a graph structure (where the
arcs correspond to the adjacency relation), and a metric structure
(where the distance is defined as the distance on the graph).

In order to compute the sensitivity bounds in a compositional way, we
need to cope with different structures at the intermediate steps, and
with different notions of distance. Consequently, we need to extend
the notions of differential privacy and sensitivity to general metric
domains.

We start by defining the notions of distance that we will need.

\begin{definition}[Hamming distance $d_H$]
  The distance between two relations $R,R' \in \mathcal{R}$ is the
  Hamming distance $d_H(R,R') = |R\;\ominus\; R'|$, the cardinality of
  the symmetric difference between $R$ and $R'$. The symmetric
  difference is defined as $R \;\ominus \;R' = (R \setminus R')\cup (R'
  \setminus R)$.
\end{definition}
 
Note that $d_H$ coincides with the graph-theoretic distance on the
graph induced by the adjacency relation $\sim$, and that $d_H(R,R')=1
\Leftrightarrow R \sim R'$.
We now extend the Hamming distance to tuples of relations, to deal with
n-ary operators.

\begin{definition}[Distance $d_{nH}$]
  The distance $d_{nH}$ between two tuples of $n$ relations
  $(R_1,\ldots,R_n),$\\ $(R'_1,\ldots,R'_n) \in \mathcal{R}^n$ is defined as:
  $d_{nH}((R_1,\ldots,R_n),(R'_1,\ldots,R'_n)) = max(d_H(R_1,R'_1), \ldots, d_H(R_n,R'_n))$
\end{definition}

Note that $d_{nH}$ coincides with the Hamming distance for $n=1$. We
chose this maximum metric instead of other distances because it allows
us to compute the sensitivity compositionally, while this is not the
case for other notions of distance. We can show counterexamples, for
instance, for both the Euclidian and the Manhattan distances.

\begin{definition}[Distance $d_E$]
  The distance between two real numbers $x,x' \in \mathbb{R}$ is the usual
  euclidean distance $d_E(x,x') = |x-x'|$.
\end{definition}

In summary, we have two metric spaces over which the relational
algebra operators work, namely $(\mathcal{R}^n, d_{nH})$, and
$(\mathbb{R},d_E)$.

\begin{example}
  Consider a relation $R$ and two tuples $\tau,\pi$ such
  that $\tau\not \in R$ and $\pi\in R$. We define its neighbors $R^+$
  and $R^{\pm}$, obtained by adding one record, and by changing one
  record, respectively:
  $$R^+ = R \cup \{\tau\} \qquad \qquad R^{\pm} = R \cup \{\tau\} \setminus \{\pi\}$$
  Their distance from $R$ is : $d_H(R,R^+) = |R \;\ominus\; R^+| = 1$, 
  and $d_H(R,R^\pm) = |R \;\ominus\; R^{\pm}| = 2$. Note also that $R \sim R^+$.
\end{example}

\begin{notation}\label{notation}
  In the following, we will use the notation $R^+$ to denote $R\cup
  \{\tau\}$ for a generic tuple $\tau$, with the assumption (unless
  otherwise specified) that $\tau\not\in R$.
\end{notation}

We now adapt the definition of differential privacy to arbitrary
metric spaces $(X,d)$ (where $X$ is the support set and $d$ the
distance function).

\begin{definition}[Differential privacy extended]\label{def:extdp}
  A randomized mechanism $\cal K: X \rightarrow Z$ on a generic metric
  space $(X,d)$ provides $\epsilon$-differential privacy if for any
  $x,x'\in X$, and any set of possible outputs $Y \subseteq Z$,
  $$ \mathit{Pr}[{\cal K}(x) \in Y] \;\leq\; \mathit{Pr}[{\cal K}(x') \in Y] 
  \cdot e^{\epsilon \cdot d(x, x')} $$
\end{definition}

It can easily be shown that Definitions \ref{def:extdp} and
\ref{def:dp} are equivalent if $d=d_H$.

We now define the sensitivity of a function on a generic metric space. 

\begin{definition}[Sensitivity extended]\label{def:sensitivity_extended}
  Let $(X,d_X)$ and $(Y,d_Y)$ be metric spaces. 
  The sensitivity $\Delta_f$ of a function $f: (X,d_X) \rightarrow (Y,d_Y)$ is defined as
  $$ \Delta_f \;=\; 
  \sup_{\begin{array}{c}
      {\scriptstyle x,x' \in X}\\[-1ex]
      { \scriptstyle x \neq x'}
    \end{array}} 
  \frac{d_Y(f(x),f(x'))}{d_X(x,x')} $$
\end{definition}

Again, we can show that Definitions~\ref{def:sensitivity_extended} and
\ref{def:sensitivity} are equivalent if $d_X=d_{nH}$ (proof in 
full version \cite{self}). This more general definition makes clear
that the sensitivity of a function is a measure of how much it
increases distances from its inputs to its outputs.

As a refinement of the definition of sensitivity, we may notice that
this attribute does not depend on the function alone, but also on the
domain, where the choice of $x,x'$ ranges to compute the supremum. In
our framework this is particularly useful because we have a very
precise description of the restrictions on the domain of an operator,
thanks to its input constrained schema
(Def~\ref{def:constrained_schema}).

\begin{definition}[Sensitivity constrained]\label{def:sensitivity_constrained}
  Given a function $f:(X,d_X) \rightarrow (Y,d_Y)$, and a set of
  constraints $C$ on $X$, the sensitivity of $f$ with respect to $C$ is
  defined as
  \[ \Delta_f(C) \;=\; 
  \sup_{\begin{array}{c}
      {\scriptstyle x,x' \in \, sol(C)}\\[-1ex]
      {\scriptstyle x \neq x'}
    \end{array}} \;
  \frac{d_Y(f(x), f(x'))}{d_X(x,x')} \]
\end{definition}

The introduction of constraints, in addition to an improved precision,
allows us to define conveniently function composition. It should be
noted that when combining two functions $f \circ g$, where $g:(Y,d_Y)
\rightarrow (Z,d_Z)$, the domain of $g$ actually depends on the
restrictions introduced by $f$ and we can take this into account
maximizing over $y,y' \in sol(C \otimes C_f)$, that is the domain
obtained combining the initial constraint $C$ and the constraint
introduced by $f$.

\section{Operators}
We now proceed to compute a bound on the sensitivity of each
relational algebra operator through a static analysis that depends
only on the relation schema the operator is applied to, and not on its
particular instances.

From a static point of view each operator will be considered as a
transformation from schema to schema (instead of a transformation from
relations to relations): they may add or remove attributes, and modify
constraints.

The following analysis is split in operators ${\tt op}:
(\mathcal{R}^n,d_{nH}) \rightarrow (\mathcal{R},d_H)$, with $n$ equals
1 or 2, and aggregation $\gamma_{f}: (\mathcal{R},d_{H}) \rightarrow
(\mathbb{R},d_{E})$. In the sensitivity analysis of the formers, given
they work only on Hamming metrics, we are only interested in their
effect on the number of rows. In our particular case, these relational
algebra operators treats all rows equally, without considering their
content. This simplification grants us the following property:

\begin{proposition}
If ${\tt op}: (\mathcal{R},d_{H}) \rightarrow (\mathcal{R},d_H)$ and
$C$ is an arbitrary set of constraints
$$ \Delta_{\tt op} (C) 
= 
\sup_{\begin{array}{c}
    {\scriptstyle R,R' \in \, \mathcal{R}(C)}\\[-1ex]
    {\scriptstyle R \neq R'}
  \end{array}} \;
\frac{d_H({\tt op}(R), {\tt op}(R'))}{d_{H}(R,R')}
=
\min \left( \Delta_{\tt op}(\emptyset), diam(C \otimes C_{\tt op}) \right) $$ 
\end{proposition}

(The proposition holds analogously for the binary case).
This property, that does not hold for general functions, allows us in
the case of relational algebra to decouple the computation of
sensitivity from the constraint system, and solve them
separately. $\Delta_{\tt op}(\emptyset)$ (from now on just
$\Delta_{\tt op}$) can be seen as the sensitivity intrinsic to each
operator, the maximum value of sensitivity the operator can cause,
when the constraints are loose enough \footnote{for all possible
domains the sensitivity can't be greater.} to be omitted.
While $diam(C \otimes C_{\tt op})$, the diameter of the co-domain of
the operator, limits the maximum distance the operator can produce,
that is the numerator in the distances ratio.

\section{Row operators}~\label{sec:row-op}
In this section we consider a first group of operators of type
$(\mathcal{R}^n,d_{nH}) \rightarrow (\mathcal{R},d_H)$ with $n=1,2$,
which are characterized by the fact that they can only add or remove
tuples, not modify their attributes.
Indeed the header of the resulting relation maintains the same set of
attributes and only the relative constraints may be modified.

\subsection{Union $\cup$}
The union of two relation is the set theoretic union of two set of
tuples with the same attributes. The example below illustrates this
operation:

\begin{center}
{\footnotesize
  \begin{tabular}{lll}
    Name    & Age & Height \\ \hline
    John    & 30  & 180    \\
    Tim     & 10  & 100    \\
  \end{tabular}
  \;$\bigcup$\;
  \begin{tabular}{lll}
    Name    & Age & Height \\ \hline
    Alice   & 45  & 160    \\
    Tim     & 10  & 100    \\
  \end{tabular}
  \;=\;
  \begin{tabular}{lll}
    Name    & Age & Height \\ \hline
    John    & 30  & 180    \\
    Tim     & 10  & 100    \\
    Alice   & 45  & 160    \\
  \end{tabular}}
\end{center}

The union of two relations may reduce their distance, leave it
unchanged or in the worst case it could double it, so the sensitivity
of union is 2.

\begin{proposition}\label{prop:union}
  The union has sensitivity $2$: $\Delta_\cup=2$.
\end{proposition}

\begin{proof}
  If $d_{2H}((R_1,R_2),(R_3,R_4))=1$ then we have two cases
  \begin{description}
  \item[a)]
    $R_3=R_1^+,R_4=R_2$ or $R_3=R_1,R_4=R_2^+$. For the symmetry of
    distance only one case needs to be considered:
    \[|(R_1 \cup R_2) \;\ominus\; (R_1^+ \cup R_2)| = \left\{
      \begin{array}{ll}
        0 & \tau \in R_2 \\
        1 & o.w.
      \end{array}
    \right.\]

    The only difference is the tuple $\tau$. If $\tau \in R_2$ then $\tau$
    would be in both results, leading to identical relations, thus
    reducing the distance to zero. If $\tau \not\in R_2$ then $\tau$ will
    again be the only difference between the results, thus resulting into
    distance $1$.

  \item[b)] $R_3=R_1^+,R_4=R_2^+$
    \[|(R_1 \cup R_2) \;\ominus\; (R_1^+ \cup R_2^+)| = \left\{
      \begin{array}{ll}
        0 & \tau_1 \in R_2 \wedge \tau_2 \in R_1 \\
        1 & \tau_1 \in R_2 \vee \tau_2 \in R_1 \\
        2 & \tau_1 \notin R_2 \wedge \tau_2 \notin R_1 \\
      \end{array}
    \right.\]
  \end{description}

  In this case we have two records differing, $\tau_1$ and $\tau_2$,
  and in the worst case they may remain different in the results,
  giving a final sensitivity of 2 for the operator.
\end{proof}

\begin{definition}[Constraints for union]
  Let $schema(R_1)=(A,C_1)$ and $schema(R_2)=(A,C_2)$. Then\\
  $ schema(R_1 \cup R_2) = (A,C_1 \vee C_2).$
\end{definition}

\subsection{Intersection\; $\cap$}
The intersection of two relation is the set theoretic intersection of
two set of tuples with the same attributes.

As for the union, the intersection applied to arguments at distance $1$
may result in a distance $0$, $1$ or $2$.

\begin{proposition}
  The intersection has sensitivity $2$: $\Delta_\cap = 2$.
\end{proposition}
\begin{proof}
  Similar to the case of Proposition~\ref{prop:union}.
\end{proof}

\begin{definition}[Constraints for intersection]
  Let $schema(R_1)=(A,C_1)$ and $schema(R_2)=(A,C_2)$.\\
  Then 
  $ schema(R_1 \cap R_2) = (a,C_1 \wedge C_2). $
\end{definition}

\subsection{Difference $\setminus$}
The difference of two relation is the set theoretic difference of two
set of tuples with the same attributes.

As in the case of the union, the difference applied to arguments at distance $1$
may result in a distance $0$, 1 or 2.

\begin{proposition}
  The set difference has sensitivity $2$: $\Delta_\setminus = 2$.
\end{proposition}
\begin{proof}
  Similar to the case of Proposition~\ref{prop:union}.
\end{proof}

\begin{definition}[Constraints for set difference]
  Let $schema(R_1)=(A,C_1)$ and $schema(R_2)=(A,C_2)$. Then
  $schema(R_1 \setminus R_2) = (A, C_1 \wedge (\neg C_2)).$
\end{definition}

\subsection{Restriction $\sigma$}\label{sec:restriction}
The restriction operator $\sigma_{\varphi}(R)$ removes all rows not
satisfying the condition $\varphi$ (typically constructed using the
predicates $=,\neq,<,>$ and the logical connectives
$\vee,\wedge,\neg$), over a subset of $R$ attributes.

As an example, consider the following SQL program that removes all
people whose age is smaller than $20$ or whose height is greater than
$180$. The table illustrates an example of application of the
corresponding restriction $\sigma_{Age \geq 20 \wedge Height<180}$.

\vskip 1em
\begin{center}
  \begin{minipage}{.5\linewidth}
    \ttfamily \small
    SELECT *\\
    FROM  ~~R\\
    WHERE ~Age>=20 AND Height<=180
\end{minipage}
\vskip 1em
\begin{minipage}{.8\linewidth}
  $\sigma_{Age \geq 20 \wedge Height<180}$ 
  {\footnotesize $\left(
      \begin{tabular}{lll}
        Name    & Age & Height \\ \hline
        John    & 30  & 180    \\
        Tim     & 10  & 100    \\
        Alice   & 45  & 160    \\
        Natalie & 20  & 175    \\
      \end{tabular}
    \right)$ =
    \begin{tabular}{lll}
      Name    & Age & Height \\ \hline
      Alice   & 45  & 160    \\
      Natalie & 20  & 175    \\
    \end{tabular}
  } \\[2ex]
\end{minipage}
\end{center} 

The restriction can be expressed in terms of set difference: 
$\sigma_{\varphi}(R) = R \setminus \{\tau \;|\; \neg \varphi(\tau)\}$. 
However the sensitivity is different because the operator is unary,
the second argument is fixed by the condition $\varphi$

\begin{proposition}
  The restriction has sensitivity $2$: $\Delta_{\sigma_{\varphi}} = 1$.
\end{proposition}

\begin{definition}[Constraints for restriction]
  Let $schema(R)=(A,C)$ and $A' \subseteq A$. Then define\\
  $schema(\sigma_{\varphi(A')}(R)) = (A,C \wedge \varphi(A')).$
\end{definition}

\section{Attribute operators}\label{sec:att-op}
The following set of operators, unlike those analyzed so far, can
affect the number of tuples of a relation, as well as its attributes.

\subsection{Projection $\pi$} 
The projection operator $\pi_{a_1, \ldots,a_n}(R)$ eliminates the
columns of $R$ with attributes other than $a_1, \ldots,a_n$, and then
deletes possible duplicates, thus reducing distances or leaving them
unchanged. It is the opposite of the restricted Cartesian product
$\times_1$ which will be presented later.

The following example illustrates the use of the projection. Here, the
attribute to preserve are {\tt Name} and {\tt Age}.

\vskip 1em
\makebox[\linewidth][c]{
  \begin{minipage}{.3\linewidth}
    \tt{SELECT Name,Age\\
      FROM  ~ R}
  \end{minipage}
\qquad
  \begin{minipage}{.55\linewidth}
    {\footnotesize
      $\pi_{Name,Age}$
      $\left(
        \begin{tabular}{lll}
          Name  & Age & Car \\ \hline
          John  & 30  & Ford    \\
          John  & 30  & Renault    \\
          Alice & 45  & Fiat    \\
        \end{tabular}
      \right)$ =
      \begin{tabular}{ll}
        Name  & Age          \\ \hline
        John  & 30           \\
        Alice & 45           \\
      \end{tabular}
    }\\[2ex]
  \end{minipage}
}

\begin{proposition}
  The projection has sensitivity $1$: $\Delta_\pi=1$.
\end{proposition}

\begin{proof}
$\qquad |\pi_{a_1,\ldots, a_n}(R) \;\ominus\; \pi_{a_1,\ldots, a_n}(R^+)| = \left\{
    \begin{array}{ll}
      0 & \exists \rho \in R.\; \forall i\in\{1,\ldots, n\}\;
      \rho(a_i) = \tau(a_i) \\
      1 & o.w.
    \end{array}
  \right. $ 
\end{proof}

\begin{definition}[Constraints for projection]
  Let $schema(R)=(A,C)$ and $A' \subseteq A$. Then\\ 
  $schema(\pi_{A'}(R)) = (A',C).$
\end{definition}

\subsection{Cartesian product}
The Cartesian product of two relation is the set theoretic Cartesian
product of two set of tuples with different attributes, with the
exception that in relations the order of attributes does not count,
thus making the operation commutative. The following example
illustrate this operation.

\begin{center}
{\footnotesize
  \begin{tabular}{lll}
    Name  & Age & Height                \\ \hline
    John  & 30  & 180                   \\
    Alice & 45  & 160                   \\
  \end{tabular}
  $~\times~$
  \begin{tabular}{ll}
    Car   & Owner                       \\ \hline
    Fiat  & Alice                       \\
    Ford  & Alice                       \\
  \end{tabular}
  $~=~$
  \begin{tabular}{lllll}
    Name  & Age & Height & Car  & Owner \\ \hline
    John  & 30  & 180    & Fiat & Alice \\
    John  & 30  & 180    & Ford & Alice \\
    Alice & 45  & 160    & Fiat & Alice \\
    Alice & 45  & 160    & Ford & Alice \\
  \end{tabular}
}\\[2ex]
\end{center}

This operator may seem odd in the context of a query language, but
it is in fact the base of the \emph{join}, the operator to merge the information of two relations.
$$ R \mathop{\bowtie}_{R.a_i=T.a_i} T = \sigma_{R.a_i=T.a_i}(R \times T) $$

We analyze now the sensitivity of the Cartesian product. 

\paragraph{One record $\times_1$} 
We first consider a restricted version $\times_1$, where on one side
we have a single tuple.

\begin{proposition}
  The operator $\times_1$ has sensitivity $1$: $\Delta_{\times_1}=1$.
\end{proposition}

\paragraph{N records $\times$}
We consider now the full Cartesian product operator. It is immediate
to see that a difference of a single row can be expanded to an
arbitrary number of records, thus causing and unbounded sensitivity.

\begin{proposition} 
  The (unrestricted) Cartesian product has unbounded sensitivity.
\end{proposition}

We now define how constraints propagate through Cartesian product:
\begin{definition}[Constraints for product]
  Let $schema(R_1)=(A_1,C_1)$ and $schema(R_2)=(A_2,C_2)$. Then
  $schema(R_1 \times R_2) = (A_1 \cup A_2, C_1 \wedge C_2).$
\end{definition}

\subsection{Restricted $\times$}\label{restricted_times}
The effect of Cartesian product is to expand each record with a block of
records, a behavior clearly against our objective of distance-preserving
computations. However we propose some restricted versions of the
operator in order to maintain its functionality to a certain extent:
\begin{itemize}
\item $\times_n$: product with blocks of a fixed $n$ size, to obtain $n$
  sensitivity. In this case $n$ representative elements can be chosen
  from the relation, the definition of policies to pick these
  elements is left to future developments.
\item $\times_{\gamma}$: a new single record is built as an
  aggregation of the relation, through the operator $\emptyset \gamma
  f$ (presented later), thus falling in the case of $\times_1$ sensitivity.
\item a mix the two approaches could be considered, building $n$
  aggregations, possibly using the operator $_{\{a_i\}}\gamma_{f}$ (presented later).
\end{itemize}

In both approaches the rest of the query can help to select the right
records from the block, for example an external restriction could be
anticipated.

\section{Aggregation $\gamma$}\label{sec:agg}
The classical relational algebra operator for aggregation 
$_{\{a_1,\ldots, a_m\}} \;\gamma\; _{\{f_1, \ldots, f_k\}} (R)$ performs the
following steps:
\begin{itemize}
  \setlength{\itemsep}{1pt}\setlength{\parskip}{0pt}\setlength{\parsep}{0pt} 
\item it partitions $R$, so that each group has all the tuples with
  the same values for each $a_i$,
\item it computes all $f_i$ for each group,
\item it returns a single tuple for each group, with the values of
  $a_i$ and of $f_i$.
\end{itemize}
The most common function founds on RDBMS are
\texttt{count,max,min,avg,sum} and we will restrict our analysis to
these ones.
The following example illustrates how we can use an aggregation
operator to know, for each type of {\tt Car}, how many people own it
and what is their average height.

\begin{center}
  \begin{minipage}{.5\linewidth}
    {\small \texttt{SELECT ~Car, Count(*), Avg(Height) \\
      FROM ~~~R \\
      GROUPBY Car}}
  \end{minipage}
\vskip 1em
  $_{\{Car\}} \gamma_{\{Count, Avg(Height)\}} 
  {\footnotesize
  \left(
  \begin{tabular}{lllll}
    Name    & Age   & Height & Car     \\ \hline
    Alice   & 45    & 160    & Ford    \\
    John    & 30    & 180    & Fiat    \\
    Frank   & 45    & 165    & Renault \\
    Natalie & 20    & 170    & Ford    \\
  \end{tabular}
  \right)
  =
  \begin{tabular}{lllll}
    Car     & Count & Avg(Height)      \\ \hline
    Ford    & 2     & 165              \\
    Fiat    & 1     & 180              \\
    Renault & 1     & 165              \\
  \end{tabular}}
$
\end{center}
\vskip .5 em

In the domain of differential privacy special care must be taken when
dealing with this operator as it is in fact the point of the query in
which our analysis of sensitivity ends and the noise must be added to
the result of the function application.

A differentially private query should return a single value, in our
case in $\mathbb{R}$, and the only queries that statically guarantee
this property are those ending with the operator $_{\emptyset}
\gamma_{f}: (\mathcal{R},d_H) \rightarrow (\mathbb{R},d_E)$ (from here
on abbreviated $\gamma _{f}$), that apply only one function $f$ to the
whole relation without grouping.
For this reason we will ignore grouping for now, and focus on queries
of the form $_{\emptyset}\gamma_{f} (Q)$ where $Q$ is a sub-query
without aggregations. It is however possible to recover the original
$_A\gamma_F$ behavior and use it in sub-queries.

\subsection{Functions}\label{sec:aggregation_functions}

In this section we analyze the sensitivity of the common mathematical
functions \texttt{count,sum, max,min} and {\tt avg}. The application
of functions coincide with the change of domain, in fact they take as
input a relation in $(\mathcal{R},d_H)$ and return a single number in
$(\mathbb{R},d_E)$, (not to be confused with a relation with a single
tuple, which also contains a single value).

Extending standard results \cite{Dwork:06:ICALP}, we can prove that,
when $f = $ \texttt{count,sum,max,min,avg} then $\Delta_f(C)$ can be
computed as follows:

\begin{proposition}\label{prop:aggregation_functions}
\ \\
  \begin{minipage}{.55\linewidth}
    $\begin{array}{rclcl}
      \Delta_{{\tt count}}(C)     & = & 1                                          \\[1ex]
      \Delta_{{\tt sum}_{a_i}}(C) & = & \max \{ |\sup(C,{a_i})|, |\inf(C,{a_i})|\} \\[1ex]
    \end{array}$
  \end{minipage}
  \begin{minipage}{.5\linewidth}
    $\begin{array}{rclcl}
      \Delta_{{\tt max}_{a_i}}(C) & = & |\sup(C,a_i)- \inf(C,a_i)|                 \\[1ex]
      \Delta_{{\tt min}_{a_i}}(C) & = & |\sup(C,a_i)- \inf(C,a_i)|                 \\[1ex]
      \Delta_{{\tt avg}_{a_i}}(C) & = & \frac{\displaystyle |\sup(C,a_i) - \inf(C,a_i)| }{\displaystyle 2}
    \end{array}$
  \end{minipage}

\end{proposition}

\subsection{Exploiting the constraint system}

The sensitivity of aggregation functions, as shown above, depends on
the range of the values of an attribute, so clearly it is important to
compute the range as accurately as possible.

The usual approach is to consider the bounds given by the domain of
each attribute. In terms of constraint system, this corresponds to
consider the solutions of the constraint $C_I = a_1\in D_1\wedge
a_2\in D_2\wedge \ldots \wedge a_n\in D_n$. I.e. the standard
approach computes the sensitivity of aggregation functions for an
attribute $a$ on the basis of $\sup(C_I,a)$ and $\inf(C_I,a)$.

In our proposal we also use $C_I$: for us it is the initial
constraint, at the beginning of the analysis of the query. The
difference is that our approach updates this constraints with
information provided by the various components of the query, and then
exploits this information to compute more accurate ranges for each
attribute. The following example illustrates the idea.

\begin{example}

Assume that $schema(R)=(\{Weight,Height\},C_I)$, and that the domain
for $Weight$ is $[0,150]$ and for $Height$ is $[0,200]$.
The following query asks the average weight of all the individuals
whose weight is below the height minus $100$.
$$\gamma_{avg(Weight)}(\sigma_{Weight \leq Height -100}(R))$$
Below we show the initial constraint $C_I$ and the constraints $C_Q$
computed by taking into account the condition of $\sigma$. Compare the
sensitivity computed using $C_I$ with the one computed using $C_Q$:
They differ because in $C_Q$ the max value of $Weight$ is $100$, while
in $C_I$ is $150$.
\[ \footnotesize 
\begin{array}{llllllll}
  C_I = \{  {\small W }\in [0,150] \;\wedge\; H \in [0,200] \} &&
  \Delta(C_I,\gamma_{avg(W)}) = \frac{|max(C_I,W) - min(C_I,W)|}{2}
  =75 \\[1ex]

  C_Q = \{ W \in [0,150] \;\wedge\; H \in [0,200] \; \wedge \; W \leq H -100 \}&&
  \Delta(C_Q,\gamma_{avg(W)})= \frac{|max(C_Q,W) - min(C_Q,W)|}{2}
  =50 
\end{array}\]
\end{example}

Hence exploiting the constraints generated by the query can lead to a
significant reduction of the sensitivity.

\subsection{Constraints generated by the functions}
We now define how to add new constraints for the newly created
attributes computed by the functions. 

\begin{definition}[Constraints for functions]
  Let $schema(R)=(A,C)$, $A' \subseteq A$ and \\
  $F = \{f_1(a_1), \ldots, f_n(a_n)\}$, where
  $a_1,\ldots,a_n \in A$. Then $schema(_{A'} \;\gamma\; _F (R)) = 
  (A' \cup \{a_{f_1} \ldots a_{f_n}\},~ C \wedge c_{f_1} \wedge \ldots \wedge c_{f_n})$, where:
  \[\begin{array}{c}
    c_{f_i} = \left\{
      \begin{array}{rcll}
        min(C,a_i) \leq & a_{f_i} & \leq sup(C,a_i) & 
        \text{if} ~ f_i = max/min/avg \\
        0  \leq         & a_{f_i} &                 & 
        \text{if} ~ f_i = sum/count   \\
      \end{array}
    \right.
  \end{array}\]
\end{definition}

\section{Global sensitivity}\label{sec:glob-sen}

We have concluded the analysis for all operators of relational
algebra, and we now define the sensitivity of the whole query in a
compositional way. 

For the computation of the sensitivity, we need to take into account
the constraint generated by it. We start by showing how to compute
this constraint, in the obvious (compositional) way.  Remember that we
have already defined the constraints generated by each relational
algebra operator in Sections~\ref{sec:row-op}, \ref{sec:att-op} and
\ref{sec:agg}.

\begin{definition}[Constraint generated by an intermediate query]
  The global constraint generated by an intermediate query $Q$ on
  relations with relational schema $r(a_1:D_1,a_2,D_2,\ldots,a_n:D_n)$
  is defined statically as:
  $$C_Q = schema(Q(R))$$
  where $R$ is any relation such that $schema(R) =
  (\{a_1,a_2,\ldots,a_n\}, C_I)$, with $C_I=a_1\in D_1\wedge a_2\in
  D_2\wedge \ldots\wedge a_n\in D_n$.
\end{definition}

We assume, the top-level operator in a query is an aggregation
$\gamma_{f}$, followed by a query composed freely using the other
operators. We now show how to compute the sensitivity of the
latter. Since it is a recursive definition, for the sake of elegance
we will assume an identity query $\mathit{Id}$. 

\begin{definition}[Intermediate query sensitivity]\label{def:query_sensitivity}
  Assume ${\tt op}: (\mathcal{R}^n,d_{nH}) \rightarrow
  (\mathcal{R},d_{H})$ and $C_{\tt op}$ the constraint obtained
  \emph{after} the application of ${\tt op}$:
  \[\begin{array}{lcll}
    S(\mathit{Id})            & = & \min(1,diam(C_{Id}))                                            & \text{base case} \\
    S({\tt op} \circ Q)       & = & \min \left( \Delta_{\tt op} \cdot S(Q), diam(C_{{\tt op} \circ Q}) \right) & \text{if}~n=1     \\
    S({\tt op } \circ (Q_1,Q_2)) & = & \min\left(\Delta_{\tt op} \cdot \max(S(Q_1),S(Q_2)), diam(C_{{\tt op} \circ (Q_1,Q_2)})\right) & \text{if}~n=2
  \end{array}\]
  where {\tt op} can be any of $\cup, \cap, \setminus, \sigma, \pi,
  \times, \times_1$ and the (classic) $_A\gamma_F$.
\end{definition}

We are now ready to define the global sensitivity of the query:

\begin{definition}[global sensitivity]\label{def:gs} 
The global sensitivity $GS$ of a query $\gamma_{f} (Q)$ is defined as:
\[\begin{array}{lcll}
  GS(\gamma_{f} (Q))           & = & \left\{ 
    \begin{array}{ll}
      \Delta_f(C_Q) \cdot S(Q) & \quad \text{if } f = \tt{count,sum,avg} \\[2ex]
      \Delta_f(C_Q)            & \quad \text{if } f = \tt{max,min}
    \end{array}
  \right.
\end{array}
\]
\end{definition}
 
The following theorem, (proof in full version \cite{self}), expresses
the soundness and the strictness of the bound computed with our
method.

\begin{theorem}[Soundness and strictness]
  The sensitivity bound computed by $GS(\cdot)$ is sound and strict. Namely:
  $$ GS(\gamma_{f} (Q))\;=\; \Delta_{\gamma_{f} (Q)} $$
\end{theorem}

\section{Related Work}\label{sec:related}
The field of privacy in statistical databases has often been
characterized by ad-hoc solutions or algorithms to solve specific
cases \cite{DBLP:dblp_conf/tamc/Dwork08}. In recent years however
there have been several efforts to develop a general framework to
define differentially private mechanisms. In the work
\cite{DBLP:dblp_conf/icfp/ReedP10} the authors have proposed a
functional query language equipped with a type system that guarantees
differential privacy. Their approach is very elegant, and based on
deep logical principles. However, it may be a bit far from the
practices of the database community, addressing which is the aim of
our paper.

The work that is closest to ours, is the PINQ framework
\cite{DBLP:dblp_conf/sigmod/McSherry09}, where McSherry extends the
LINQ language, with differential privacy functionalities developed by
himself, Dwork and others in \cite{DBLP:dblp_conf/tcc/DworkMNS06}.

Despite this existing implementation we felt the need for a more
universal language to explore our ideas, and the mathematically-based
framework of relational algebra seemed a natural choice. Furthermore
the use of a constraint system to increase the precision of the
sensitivity bound was, to out knowledge, never explored before.

\section{Conclusions and future work}\label{sec:conclusion}
We showed how a classical language like relational algebra can be a
suitable framework for differential privacy and how technology already
in place, like check constraints, can be exploited to improve the
precision of our sensitivity bounds.

Our analysis showed how the most common operation on databases, the
join $\bowtie$, poses great privacy problems and in future we hope to
develop solutions to this issue, possibly along the lines already
presented in Section~\ref{restricted_times}.

In this paper we have considered only the sensitivity, that is the
effect on distances of operators, while another interesting aspect
would be to compute the effect on the $\epsilon$ exponent as explored
in \cite{DBLP:dblp_conf/sigmod/McSherry09}, and possibly propose
convenient strategies to query as much as possible over disjoint data
sets.

\bibliographystyle{eptcs}
\bibliography{biblio}

\end{document}

%% file: style.tex
\usepackage[utf8x]{inputenc}   
\usepackage[pdftex]{xcolor}
\usepackage{hyperref}          
\usepackage{listings}          
\usepackage{array}             
\usepackage{multicol}          

\usepackage[pdftex]{graphicx}  
\usepackage{wrapfig}           
\usepackage{subfig}            
\usepackage{tikz}              

\usepackage{amsfonts}
\usepackage{amsmath}
\usepackage{amssymb} 
\usepackage{mathtools}  
\usepackage{framed}	
\usepackage[hyperref, amsmath, thref, amsthm, framed, thmmarks]{ntheorem} 
\usepackage{bussproofs}	
\usepackage{cancel}     
\usepackage{extarrows}  
\usepackage{stmaryrd}   
\usepackage{dsfont}     

\definecolor{lightgray}{RGB}{240,240,240}

\theoremstyle{plain}
\theoremheaderfont{\normalfont\bfseries}
\theorembodyfont{\itshape}
\theoremsymbol{}
\newtheorem{theorem}{Theorem} 

\theoremstyle{break}
\theoremheaderfont{\normalfont\bfseries}
\theorembodyfont{\normalfont}
\theoremsymbol{}

\theoremstyle{plain}
\theorembodyfont{\normalfont}
\theoremsymbol{}
\newtheorem{definition}{Definition}

\theoremstyle{plain}
\theorembodyfont{\normalfont}
\theoremsymbol{}
\newtheorem{example}{Example}

\theoremstyle{plain}
\theorembodyfont{\normalfont}
\theoremsymbol{}
\theoremseparator{\newline}
\newtheorem{proposition}{Proposition}

\theoremstyle{plain}
\theorembodyfont{\normalfont}
\theoremsymbol{}
\theoremseparator{\newline}
\newtheorem{notation}{Notation}

{\endMakeFramed}

\newcolumntype{V}{>{\centering\arraybackslash} m{.4\linewidth} }

